\documentclass[11pt]{article}
\usepackage[a4paper,margin=1in]{geometry}
\usepackage[T1]{fontenc}
\usepackage{lmodern,microtype}
\pdfmapfile{+lm.map}
\pdfmapfile{+cm.map}
\pdfmapfile{+symbols.map}
\usepackage{amsmath,amssymb,amsthm,mathtools}
\usepackage{booktabs,enumitem,needspace}

\usepackage{xcolor,tikz}
\usepackage{graphicx,flafter,float}
\usepackage{algorithm,algorithmic}
\usetikzlibrary{arrows.meta,positioning,calc}
\usepackage[hidelinks]{hyperref}
\hypersetup{pdftitle={Polynomial Kernels for Interval Completion},pdfauthor={Zimo Sheng, Tian Bai, Mingyu Xiao}}

\newtheorem{theorem}{Theorem}[section]
\newtheorem{lemma}[theorem]{Lemma}

\theoremstyle{definition}

\theoremstyle{remark}

\newcommand{\problemstatement}[4]{%
  \par\medskip
  \begingroup
  \setlength{\fboxsep}{8pt}%
  \noindent\fbox{%
    \begin{minipage}{\dimexpr\linewidth-2\fboxsep-2\fboxrule\relax}
      \textbf{#1}\par\smallskip
      \textbf{Input.} #2\par\smallskip
      \textbf{Parameter.} #3\par\smallskip
      \textbf{Question.} #4
    \end{minipage}%
  }
  \par
  \endgroup
  \medskip
}

\title{Polynomial Kernels for Interval Completion}
\author{
Zimo Sheng \qquad Tian Bai\thanks{Corresponding author.} \qquad Mingyu Xiao\\[0.5em]
{\small
\href{mailto:shengzimo2016@gmail.com}{\texttt{shengzimo2016@gmail.com}}
\quad
\href{mailto:tian.bai@uib.no}{\texttt{tian.bai@uib.no}}
\quad
\href{mailto:myxiao@gmail.com}{\texttt{myxiao@gmail.com}}}
}
\date{}

\newcommand{\ic}{\operatorname{ic}}
\newcommand{\oct}{\operatorname{oct}}
\newcommand{\forwardcolor}{\mathord{\prec}}
\newcommand{\reversecolor}{\mathord{\succ}}
\let\originaltriangleleft\triangleleft
\renewcommand{\triangleleft}{\mathrel{\originaltriangleleft}}

\newcommand{\nonedges}{\overline E}

\renewcommand{\textsc}[1]{\ifmmode\text{\normalfont\scshape #1}\else{\normalfont\scshape #1}\fi}
\begin{document}
\maketitle
\begin{abstract}
  An interval graph is the intersection graph of a family of intervals on the real line.
\textsc{Interval Completion} asks whether a given graph can be transformed into an interval graph by adding at most $k$ edges.
Although the problem is fixed-parameter tractable when parameterized by $k$, whether it admits a polynomial kernel has long been an open question.
We resolve this question by giving the first polynomial kernel for \textsc{Interval Completion}.
Our main contribution is a parameter-preserving polynomial-time reduction from \textsc{Interval Completion} to \textsc{Odd Cycle Transversal} (OCT).
Combining this reduction with the known randomized and deterministic polynomial kernels for OCT and a polynomial-time reduction back to \textsc{Interval Completion}, we obtain a randomized kernel with $\widetilde O(k^{18})$ vertices and $\widetilde O(k^{36})$ edges, and a deterministic kernel with $O(k^{36})$ vertices and $O(k^{72})$ edges.
Here, $\widetilde O$ suppresses polylogarithmic factors in $k$.
\end{abstract}

\section{Introduction}
\label{sec:ic-introduction}
\label{sec:introduction}

An interval representation of a graph assigns an interval on the real line to each vertex so that two vertices are adjacent if and only if their intervals intersect.
Graphs admitting such a representation are called interval graphs.
They arise naturally in molecular biology, where interval intersections model overlap relationships between genomic regions~\cite{FulkersonGross1965}.
Interval graphs can be recognized in $O(n+m)$ time, where $n$ and $m$ denote the numbers of vertices and edges, respectively~\cite{BoothLueker1976}.

Graph modification problems ask whether a given graph can be transformed into a member of a prescribed graph class using a bounded number of allowed operations.
These problems have been widely studied for many target graph classes~\cite{NatanzonShamirSharan2001,CrespelleEtAl2023}.
For interval graphs, natural variants allow vertex deletions, edge deletions, or edge additions.
We study the edge-addition variant, \textsc{Interval Completion}, which asks whether a given graph can be made into an interval graph by adding at most $k$ edges~\cite{RaviAgrawalKlein1991}.
The problem is also related to matrix profile minimization and graph searching~\cite{KaplanShamirTarjan1999,FominGolovach2000}.
Formally, we consider the following parameterized problem. 

\problemstatement
  {\textsc{Interval Completion}}
  {A finite undirected simple graph $G$ and a nonnegative integer $k$.}
  {The number $k$ of allowed edge additions.}
  {Is there a set $F$ of at most $k$ nonedges of $G$ such that $G+F$ is an interval graph?}

\textsc{Interval Completion} is NP-complete~\cite{KashiwabaraFujisawa1979}. 
We study the problem parameterized by the number $k$ of allowed edge additions.
For this parameterization, Kaplan, Shamir, and Tarjan~\cite{KaplanShamirTarjan1999} asked whether \textsc{Interval Completion} is fixed-parameter tractable.
Villanger et al.~\cite{VillangerEtAl2009} settled this question with an algorithm running in $O(k^{2k}n^{3}m)$ time.
Cao~\cite{Cao2016} subsequently improved the running time to the single-exponential bound $6^{k} n^{O(1)}$.
Bliznets et al.~\cite{BliznetsEtAl2018} further obtained a subexponential algorithm running in $2^{O(\sqrt{k}\log k)}n^{O(1)}$ time.

\begingroup
Whether \textsc{Interval Completion} admits a polynomial kernel parameterized by $k$ has remained a longstanding open problem.
The question was posed by Villanger et al.~\cite{VillangerEtAl2009} and raised again by Bliznets et al.~\cite{BliznetsEtAl2018}.
Crespelle et al.~\cite{CrespelleEtAl2023} describe it as a ``notoriously hard'' problem.
Jansen~\cite{Jansen2023KernelsAndBeyond,Jansen2026QuoVadis} also included it among the most prominent open problems in kernelization in his \emph{Kernels and Beyond} session at FPT Fest 2023.

Polynomial kernels were already known for several closely related completion problems.
Proper interval graphs form a subclass of interval graphs: they admit an interval representation in which no interval strictly contains another.
Bessy and Perez~\cite{BessyPerez2013} gave an $O(k^{3})$-vertex kernel for \textsc{Proper Interval Completion}.
For chordal graphs, a superclass of interval graphs, Natanzon, Shamir, and Sharan~\cite{NatanzonShamirSharan2000} obtained an $O(k^{2})$-vertex kernel for \textsc{Chordal Completion}.
Thus polynomial kernels were known for completion into both a subclass and a superclass of interval graphs, while the intermediate case remained open.
For restricted input graphs, the results of Bessy and Perez~\cite{BessyPerez2013} also yield an $O(k^{2})$-vertex kernel for \textsc{Interval Completion} on co-bipartite graphs, whose vertices can be partitioned into two cliques.

The unresolved status of \textsc{Interval Completion} also stood out in the broader landscape of graph completion problems.
As noted by Crespelle et al.~\cite{CrespelleEtAl2023}, among the hereditary graph completion problems discussed in their survey that were known to admit subexponential parameterized algorithms, \textsc{Interval Completion} was the only one whose polynomial-kernel status remained open.
\par\endgroup

\subsection{Our Results}
We resolve the longstanding open question by giving the first randomized and deterministic polynomial kernels for \textsc{Interval Completion} on general graphs.
The key ingredient is a deterministic polynomial-time reduction that preserves the parameter exactly.
The target problem is \textsc{Odd Cycle Transversal} (OCT), defined as follows.

\problemstatement
{\textsc{Odd Cycle Transversal} (OCT)}
{A finite undirected simple graph $\hat{G}$ and a nonnegative integer $k$.}
{The number $k$ of allowed vertex deletions.}
{Is there a set $D\subseteq V(\hat{G})$ with $|D|\leq k$ such that $\hat{G}-D$ is bipartite?}

We will prove the following theorem. 
\begin{theorem}
  \label{thm:ic-oct-reduction}
  \label{thm:reduction}
There is a deterministic polynomial-time algorithm that maps each \textsc{Interval Completion} instance $(G,k)$ to an equivalent OCT instance $(\hat{G},k)$, where $\hat{G}$ has $O(n^{4})$ vertices and edges and $n=|V(G)|$.
\end{theorem}

Kratsch and Wahlstr\"om~\cite{KratschWahlstrom2014} gave a randomized polynomial kernel for OCT.
{Their subsequent work~\cite{KratschWahlstrom2020} yields a randomized kernel that can be converted into an ordinary OCT instance with $\widetilde O(k^{6})$ vertices and edges.}
Very recently, Koana and Kumabe~\cite{KoanaKumabe2026} gave a deterministic polynomial kernel with $O(k^{12})$ vertices and edges.
{Combining Theorem~\ref{thm:ic-oct-reduction} with these kernels and a polynomial-time reduction back to \textsc{Interval Completion}, we obtain the following kernels.}

\begin{theorem}
  \label{thm:ic-kernel}
  \label{thm:kernel}
  \textsc{Interval Completion} admits a deterministic polynomial kernel with $O(k^{36})$ vertices, $O(k^{72})$ edges, and output parameter bounded by $O(k^{48})$.
\end{theorem}

\begin{theorem}
  \label{thm:ic-kernel-1}
  \label{thm:kernel-1}
  \textsc{Interval Completion} admits a randomized polynomial kernel with $\widetilde O(k^{18})$ vertices, $\widetilde O(k^{36})$ edges, and output parameter bounded by $\widetilde O(k^{24})$.
  The failure probability is at most $2^{-\Omega(n)}$, where $n=|V(G)|$.
  Errors are limited to false negatives.
\end{theorem}

\subsection{Technical Overview}

Our kernelization has two stages: a parameter-preserving reduction to OCT, developed in Section~\ref{sec:construction}, followed by kernelization and a reduction back to \textsc{Interval Completion}, developed in Section~\ref{sec:kernel}.

We first consider the reduction from an interval completion instance on graph $G$ to an OCT instance on graph  $\hat{G}$. 
An interval graph has no induced $C_4$. Thus, an induced $C_4$ illustrates a local obstruction to an interval completion: at least one of its two diagonals must be filled.
However, hitting all such pairs of nonedges alone does not guarantee an interval graph.
Every nonedge that remains unfilled requires its endpoint intervals to be disjoint, and the associated left-to-right choices must be compatible with a single interval representation.
For example, suppose interval $a$ lies to the left of interval $b$, and interval $b$ lies to the left of interval $c$.
Now require interval $c$ to lie to the left of interval $a$.
These three requirements cannot hold in an interval representation.

We use two colors to encode which of the two disjoint endpoint intervals lies to the left of the other.
With respect to a fixed ordering of the vertices of $G$, the color assigned to a nonedge records which endpoint interval lies to the left of the other.

We construct $\hat G$ as follows. First, we introduce one vertex of $\hat G$ for each nonedge of $G$.
To express the constraints imposed by the adjacencies in $G$, for certain pairs of vertices in $\hat G$, we add either an edge, requiring them to receive different colors, or a path of length two, requiring them to receive the same color. Each such path uses a new additional vertex.
The resulting same-color and different-color constraints are simultaneously satisfiable exactly when $\hat G$ is bipartite.
The construction has $O(n^{4})$ vertices and edges. The details are given in Section~\ref{sec3.1}.

The correctness proof establishes both directions of the correspondence.
Given an interval completion, its interval representation determines a two-coloring after we delete the vertices corresponding to the filled nonedges.
For the converse, we first show that a minimum OCT can be chosen to contain only vertices corresponding to nonedges: a deleted internal path vertex can be replaced by an endpoint without increasing the solution size.
Lemma~\ref{lem:reverse} then uses a two-coloring of the remaining graph to recursively construct an ordering of all interval endpoints.
An averaging argument bounds the number of added edges by the size of the minimum OCT.
Together, these transformations establish the parameter-preserving equivalence in Theorem~\ref{thm:reduction}.

In Section~\ref{sec:kernel}, we apply either the randomized or the deterministic OCT kernel to $(\hat G,k)$.
We then transform the resulting ordinary OCT instance into an equivalent \textsc{Interval Completion} instance.
Bounding the size and output parameter of this back reduction gives Theorems~\ref{thm:ic-kernel} and~\ref{thm:ic-kernel-1}.
Figure~\ref{fig:kernel-overview} summarizes the kernelization.

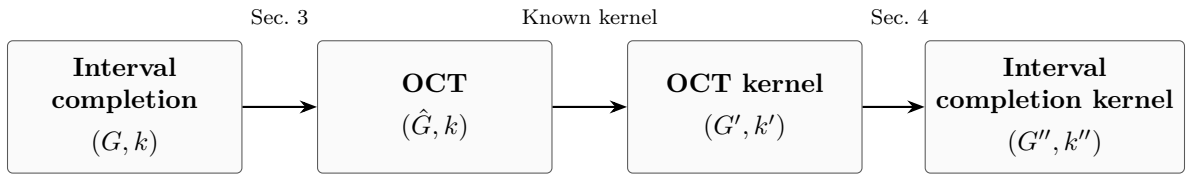
\begin{figure}[H]
  \centering
  \begingroup
  \begin{tikzpicture}[
      font=\small,text=black,
      flow/.style={-{Stealth[length=2.3mm]},draw=black,line width=.8pt},
      stage/.style={draw=black!65,fill=black!2,rounded corners=2pt,
        align=center,minimum width=3.1cm,minimum height=1.75cm,
        inner sep=4pt}]
    \node[stage] (input) at (0,0)
      {\textbf{Interval}\\\textbf{completion}\\[4pt]$(G,k)$};
    \node[stage] (oct) at (4.1,0)
      {\textbf{OCT}\\[4pt]$(\hat{G},k)$};
    \node[stage] (smalloct) at (8.2,0)
      {\textbf{OCT kernel}\\[4pt]$(G^{\prime},k^{\prime})$};
    \node[stage] (output) at (12.3,0)
      {\textbf{Interval}\\\textbf{completion kernel}\\[4pt]
       $(G^{\prime\prime},k^{\prime\prime})$};
    \draw[flow] (input.east)--(oct.west);
    \draw[flow] (oct.east)--(smalloct.west);
    \draw[flow] (smalloct.east)--(output.west);
    \node[font=\scriptsize] at (2.05,1.2)
      {Sec.~\ref{sec:construction}};
    \node[font=\scriptsize] at (6.15,1.2) {Known kernel};
    \node[font=\scriptsize] at (10.25,1.2)
      {Sec.~\ref{sec:kernel}};
  \end{tikzpicture}
  \caption{Overview of the kernelization.}
  \label{fig:kernel-overview}
  \endgroup
\end{figure}

\section{Preliminaries}
\label{sec:preliminaries}

We work on graphs that are finite, undirected, simple, and unweighted.
Unless otherwise specified, $n=|V(G)|$ and $m=|E(G)|$ denote the numbers of vertices and edges of the input graph $G$, respectively.
Fix an arbitrary labeling $V(G)=\{v_{1},v_{2},\ldots,v_{n}\}$ with the order $v_{1}\triangleleft v_{2}\triangleleft\cdots\triangleleft v_{n}$.
The relation $u\triangleleft v$ means that $u$ precedes $v$ in this fixed order.
A nonedge is a pair of distinct vertices that are not adjacent.
The set of all nonedges is denoted by $\nonedges(G)=\binom{V(G)}{2}\setminus E(G)$ and $\bar{m}=|\nonedges(G)|=\binom{n}{2}-m$.

An edge or nonedge with endpoints $u,v$ is denoted by $uv$.
For an edge or nonedge $e=uv$ of $G$, we assume $u\triangleleft v$.
We set $e_{ij}=v_{i}v_{j}$ for $i<j$.
For $u,v\in V(G)$, the relation $u\sim v$ means that $u=v$ or $uv\in E(G)$.
For a graph $G$ and $U\subseteq V(G)$, let $G[U]$ denote the subgraph induced by $U$, and let $G-U = G[V(G)\setminus U]$ be the graph obtained by deleting $U$.
Let $d_{G}(v)$ denote the degree of $v$ in $G$, and let $\overline{G}$ denote the complement of $G$.
For $F\subseteq\nonedges(G)$, the graph $G+F$ is obtained by adding the edges in $F$.

A graph $H$ is an interval graph if its vertices can be assigned closed intervals such that distinct vertices are adjacent exactly when their intervals intersect.
Such an assignment is an interval representation of $H$.
The interval assigned to $v\in V(H)$ is denoted by $I_{v}=[L_{v},R_{v}]$.
An interval supergraph $H$ of the graph $G$ is an interval graph on the same vertex set with $E(G) \subseteq E(H)$.
In an interval representation, $u\prec v$ and $I_{u}\prec I_{v}$ both mean $R_{u}<L_{v}$, so $I_{u}$ lies entirely to the left of $I_{v}$.
The relations $u\succ v$ and $I_{u}\succ I_{v}$ both mean $v\prec u$.

A set $F\subseteq\binom{V(G)}{2}\setminus E(G)$ is an \emph{interval completion} of a graph $G$ if $G+F$ is an interval graph; its minimum size is $\ic(G)$.
A set $D\subseteq V(\hat{G})$ is an \emph{odd cycle transversal} of a graph $\hat{G}$ if $\hat{G}-D$ is bipartite; its minimum size is $\oct(\hat{G})$.
A proper two-coloring of a graph $\hat{G}$ is a map $\chi:V(\hat{G})\to\{\forwardcolor,\reversecolor\}$ such that $\chi(u)\neq\chi(v)$ for every edge $uv\in E(\hat{G})$.
A graph is bipartite if and only if it admits a proper two-coloring.

Two instances are equivalent if they have the same answer.
A kernelization transforms an input instance $(G,k)$ in polynomial time into an equivalent instance $(G',k')$ of the same problem such that the size of the graph $G'$ and the parameter $k'$ are bounded by a function of $k$.
{It is a \emph{polynomial kernel} when both $|V(G')|$ and $k'$ are bounded by polynomials in $k$.}

A randomized kernelization runs in randomized polynomial time and satisfies the same size and parameter bounds for every outcome of its random choices.
An error that maps a Yes-instance to a No-instance is called a false negative.
The randomized kernel used here always maps No-instances to No-instances.
For a Yes-instance on $n$ vertices, the probability of producing a No-instance is at most $2^{-\Omega(n)}$.

\section{Reduction to the OCT problem}
\label{sec:construction}
In this section, we give the construction of an OCT instance $(\hat{G},k)$ from an instance $(G,k)$ of \textsc{Interval Completion}.
We show that $\ic(G)=\oct(\hat{G})$, so the parameter $k$ is preserved.
Recall that we fix an arbitrary order of the vertices in $G$: $v_{1}\triangleleft v_{2}\triangleleft\cdots\triangleleft v_{n}$.
When we write an edge or nonedge $uv$, we always assume $u\triangleleft v$.

\subsection{The construction}\label{sec3.1}

We first define the construction of $\hat{G}$ as follows.
\begin{enumerate}[label=\arabic*.,leftmargin=*]
  \item For each nonedge $e$ of $G$, create a \emph{variable vertex} $x_{e}$.
  \item For each pair of distinct nonedges $e_{1}=uv$ and $e_{2}=xy$ with $u\triangleleft v$ and $x\triangleleft y$, apply both of the following rules.
  \begin{enumerate}[label=(\roman*),leftmargin=*]
    \item If $u\sim x$ and $v\sim y$, create a new auxiliary vertex $w_{e_{1},e_{2}}$ adjacent only to $x_{e_{1}}$ and $x_{e_{2}}$.

    \item If $u\sim y$ and $v\sim x$, add the edge $x_{e_{1}}x_{e_{2}}$.
  \end{enumerate}
\end{enumerate}
The two conditions may both hold, in which case we add both the path and the edge.
We may also use the abbreviations $x_{ij}=x_{e_{ij}}$ and $w_{ij,pq}=w_{e_{ij},e_{pq}}$.
Figure~\ref{fig:connections} illustrates the two connection rules.

Consider an interval representation of a completion $G+F$.
Filling a nonedge $e=uv$ makes $I_{u}$ and $I_{v}$ intersect, so neither interval lies entirely to the left of the other.
We represent this by deleting $x_{e}$.
For each unfilled nonedge $e=uv$, we set $\chi(x_{e})=\forwardcolor$ if $u\prec v$ and $\chi(x_{e})=\reversecolor$ if $u\succ v$.
Thus the color of a variable vertex records whether its separation follows or reverses the fixed vertex order.
On auxiliary vertices, the same symbols serve only as color labels.
Figure~\ref{fig:connections} shows how the edges of $G$ constrain these colors.

\begin{figure}[H]
  \centering
  \begin{tikzpicture}[font=\small,>=Latex,
      v/.style={circle,draw,fill=white,inner sep=1.8pt},
      forward/.style={v,minimum size=5mm},
      reverse/.style={forward,fill=black,text=white},
      interval/.style={|-|,line width=1pt},
      every label/.style={text=black}]
    \node at (.75,2.15) {$G$};
    \node at (5.85,2.15) {Intervals};
    \node at (10.8,2.15) {$\hat{G}$};
    \foreach \shift/\kind in {0/0,-3/1}{
    \begin{scope}[yshift=\shift cm]
      \node[v,label=left:$u$] (gu) at (0,1.4) {};
      \node[v,label=left:$v$] (gv) at (0,0) {};
      \node[v,label=right:$x$] (gx) at (1.5,1.4) {};
      \node[v,label=right:$y$] (gy) at (1.5,0) {};
      \draw[dashed,gray] (gu)--node[left]{$e_{1}$}(gv) (gx)--node[right]{$e_{2}$}(gy);
      \ifnum\kind=0 \draw(gu)--(gx) (gv)--(gy);\else \draw(gu)--(gy) (gv)--(gx);\fi
      \draw[->] (2.2,.7)--(2.8,.7);
      \fill[black!4] (4.4,-.32) rectangle (5.2,1.72);
      \fill[black!4] (6.7,-.32) rectangle (7.5,1.72);
      \foreach \coord in {4.4,5.2,6.7,7.5}
        \draw[dashed,gray!35] (\coord,-.32)--(\coord,1.72);
      \foreach \name/\level in {u/1.6,v/1,x/.4,y/-.2}
        \node[anchor=east] at (3.7,\level) {$I_{\name}$};
      \draw[interval] (4,1.6)--(5.6,1.6);
      \draw[interval] (6.3,1)--(7.9,1);
      \ifnum\kind=0
        \draw[interval] (4.4,.4)--(5.2,.4);
        \draw[interval] (6.7,-.2)--(7.5,-.2);
      \else
        \draw[interval] (6.7,.4)--(7.5,.4);
        \draw[interval] (4.4,-.2)--(5.2,-.2);
      \fi
      \draw[->] (8.4,.7)--(9,.7);
      \node[forward,label=below:$x_{e_{1}}$] (xe) at (9.6,.7) {$\forwardcolor$};
      \ifnum\kind=0
        \node[reverse,label=above:$w_{e_{1},e_{2}}$] (w) at (10.85,.7) {$\reversecolor$};
        \node[forward,label=below:$x_{e_{2}}$] (xf) at (12.1,.7) {$\forwardcolor$};
        \draw(xe)--(w)--(xf);
      \else
        \node[reverse,label=below:$x_{e_{2}}$] (xf) at (12.1,.7) {$\reversecolor$};
        \draw(xe)--(xf);
      \fi
    \end{scope}}
  \end{tikzpicture}
  \caption{The two connection rules for distinct endpoints, with $I_{u}\prec I_{v}$.
  Solid edges in $G$ indicate the required adjacencies, and dashed edges mark $e_{1},e_{2}$.
  The interval placements give equal colors in the top row and different colors in the bottom row.}
  \label{fig:connections}
\end{figure}
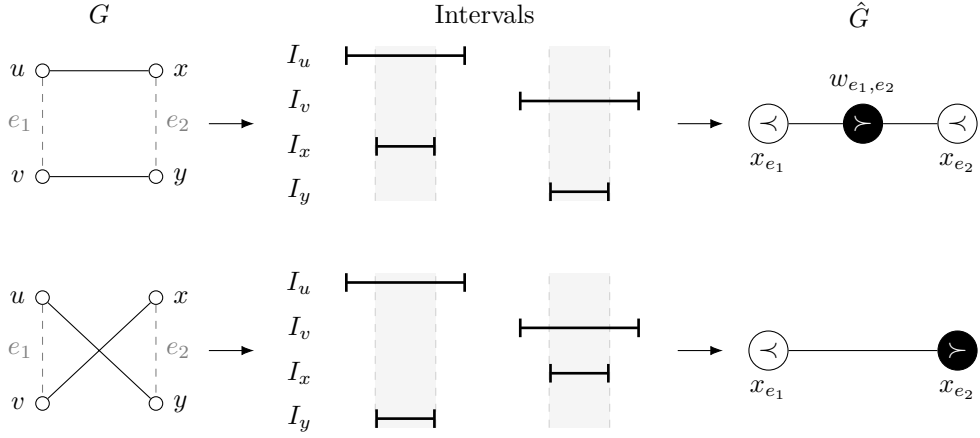

In the top row, assume $u\prec v$.
The edges $ux$ and $vy$ make $I_{x}$ intersect $I_{u}$ and $I_{y}$ intersect $I_{v}$.
Since $I_{x}$ and $I_{y}$ are disjoint, $I_{x}\prec I_{y}$, giving $x\prec y$.
The two variable vertices must therefore have the same color, as enforced by the length-two path in a proper two-coloring.
In the bottom row, the edges $uy$ and $vx$ force $x\succ y$ when $u\prec v$.
The two variable vertices must then have different colors, as enforced by the direct edge.

\begingroup
Figure~\ref{fig:c4-construction} illustrates the construction for the graph $G$, which consists of a single cycle $(v_{1},v_{2},v_{3},v_{4},v_{1})$.
There are two nonedges $e_{13}$ and $e_{24}$ in $G$.
Initially, $\hat{G}$ contains only two vertices $x_{13}$ and $x_{24}$. 
Both connection rules apply to $e_{13}$ and $e_{24}$, so $\hat G$ is the triangle on $x_{13}$, $x_{24}$, and $w_{13,24}$.

\begin{figure}[H]
  \centering
  \begingroup
  \begin{tikzpicture}[font=\small,text=black,
      v/.style={circle,draw=black,fill=white,inner sep=1.8pt},
      every label/.style={text=black}]
    \begin{scope}
      \node at (.75,2.25) {$G$};
      \node[v,label=left:$v_1$] (g1) at (0,1.4) {};
      \node[v,label=right:$v_2$] (g2) at (1.5,1.4) {};
      \node[v,label=right:$v_3$] (g3) at (1.5,0) {};
      \node[v,label=left:$v_4$] (g4) at (0,0) {};
      \draw[black] (g1)--(g2)--(g3)--(g4)--(g1);
      \draw[dashed,gray] (g1)--node[pos=.3,sloped,above,text=black]{$e_{13}$}(g3);
      \draw[dashed,gray] (g2)--node[pos=.7,sloped,above,text=black]{$e_{24}$}(g4);
    \end{scope}
    \begin{scope}[xshift=4.1cm]
      \node at (.75,2.25) {$\overline G$};
      \node[v,label=left:$v_1$] (b1) at (0,1.4) {};
      \node[v,label=right:$v_3$] (b3) at (1.5,0) {};
      \node[v,label=right:$v_2$] (b2) at (1.5,1.4) {};
      \node[v,label=left:$v_4$] (b4) at (0,0) {};
      \draw[black] (b1)--(b3) (b2)--(b4);
    \end{scope}
    \begin{scope}[xshift=8.2cm]
      \node at (.9,2.25) {$\hat G$};
      \node[v,label=below:$x_{13}$] (xe) at (0,0) {};
      \node[v,label=below:$x_{24}$] (xf) at (1.8,0) {};
      \node[v,label=above:$w_{13,24}$] (w) at (.9,1.4) {};
      \draw[black] (xe)--(w)--(xf)--(xe);
    \end{scope}
  \end{tikzpicture}
  \caption{The construction for $G=C_4$. Dashed lines mark the nonedges of $G$.}
  \label{fig:c4-construction}
  \endgroup
\end{figure}
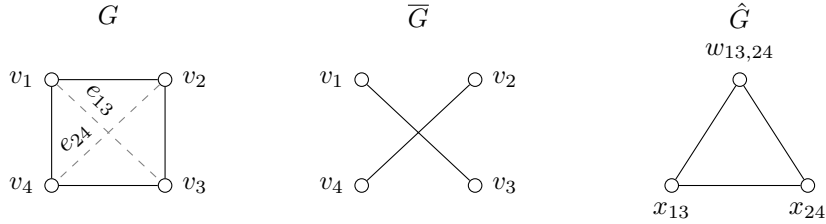

\par\endgroup

There are $\bar{m}$ variable vertices, and each pair of nonedges creates at most one auxiliary vertex and three edges.
Hence $|V(\hat{G})|\leq \bar{m}+\bar{m}(\bar{m}-1)/2$ and $|E(\hat{G})|\leq3\bar{m}(\bar{m}-1)/2$.
After an adjacency matrix of $G$ has been built, each pair of nonedges requires only a constant number of adjacency queries.
The construction therefore takes $O(n^{4})$ time.

\subsection{Correctness of the reduction}

We now prove the correctness of the reduction.
Let $D$ be a set of variable vertices of the constructed graph $\hat G$.
Assign each variable vertex outside $D$ a color $\chi(x_{e})\in\{\forwardcolor,\reversecolor\}$.
For each corresponding nonedge $e=uv$, choose $u\prec v$ if $\chi(x_{e})=\forwardcolor$ and $u\succ v$ otherwise.
In statements about selected separations, the endpoints are named in their selected directions, independently of the fixed vertex order.
The following basic property of the construction is immediate from the definitions.

\begin{lemma}
  \label{lem:color-directions}
  The auxiliary vertices can be colored to obtain a proper two-coloring of $\hat G-D$
  if and only if there are no selected separations $u\prec v$ and $x\prec y$ with
  \begin{equation*}
    u\sim y \quad\text{and}\quad x\sim v.
  \end{equation*}
\end{lemma}
\begin{proof}
  We first consider the forward implication.
  Suppose the auxiliary vertices have been colored so that $\hat G-D$ is properly two-colored.
  Consider two distinct retained nonedges $e_{1}=uv$ and $e_{2}=xy$, with $u\triangleleft v$ and $x\triangleleft y$.
  If $u\sim x$ and $v\sim y$, the auxiliary path forces their variable vertices to have the same color.
  If $u\sim y$ and $v\sim x$, the direct edge forces them to have different colors.
  Suppose their selected separations form a forbidden pair.
  If their colors differ, then $u\sim x$ and $v\sim y$, contradicting the auxiliary path.
  If their colors agree, then $u\sim y$ and $v\sim x$, contradicting the direct edge.
  Thus the condition is necessary.

  We now consider the reverse implication.
  Suppose no two selected separations satisfy the stated relations.
  Every retained direct edge then has differently colored ends, and every auxiliary vertex with two retained neighbors has equally colored neighbors.
  Give each such auxiliary vertex the opposite color.
  An auxiliary vertex with one retained neighbor receives the opposite color to that neighbor, and one with no retained neighbors may receive either color.
  There are no edges between auxiliary vertices, so this gives a proper two-coloring of $\hat{G}-D$.
\end{proof}

\begin{lemma}
  \label{lem:forward}
  $\oct(\hat{G})\leq\ic(G)$.
\end{lemma}
\begin{proof}
  Let $F$ be a minimum interval completion of $G$ and set $D=\{x_{e}:e\in F\}$.
  Take an interval representation $I_{v}=[L_{v},R_{v}]$ of $G+F$.
  Orient each nonedge outside $F$ according to its interval separation and assign the corresponding variable colors as in Lemma~\ref{lem:color-directions}.
  Suppose two selected separations $u\prec v$ and $x\prec y$ satisfy both $u\sim y$ and $x\sim v$.
  Then $I_{u}\cap I_{y}\neq\varnothing$ and $I_{x}\cap I_{v}\neq\varnothing$, giving
  \begin{equation*}
    R_{u}<L_{v}\leq R_{x}<L_{y}\leq R_{u},
  \end{equation*}
  a contradiction.
  By Lemma~\ref{lem:color-directions}, we can color the auxiliary vertices to obtain a proper two-coloring of $\hat{G}-D$.
  Thus $D$ is an OCT.
  Since $|D|=|F|=\ic(G)$, we obtain
  \begin{equation*}
    \oct(\hat{G})\leq\ic(G).\qedhere
  \end{equation*}
\end{proof}

The reverse inequality is more challenging.
One may notice that any OCT can be replaced by one containing only variable vertices without increasing its size.
Indeed, replace each auxiliary vertex $w_{e_{1},e_{2}}$ in the OCT with one of its variable neighbors.
Each restored auxiliary vertex then has at most one retained neighbor, so a proper two-coloring extends to it.
Hence a minimum OCT consisting only of variable vertices always exists.
{However}, its corresponding nonedges need not form an interval completion, as Figure~\ref{fig:oct-fill-counterexample} shows.
We instead construct an interval representation of a completion based on any minimum OCT, and use an averaging argument to show that one of the resulting completions has size at most that of the OCT.

\begin{figure}[H]
  \centering
  \begin{tikzpicture}[font=\small,
      v/.style={circle,draw,fill=white,inner sep=2pt},
      deleted/.style={v,draw=black,fill=black!10},
      added/.style={black,line width=1pt},
      auxiliary/.style={draw=black,line width=.7pt},
      deletedaux/.style={draw=black,dashed,line width=.7pt},
      smallw/.style={circle,draw=white,fill=black,
        line width=.4pt,inner sep=0pt,minimum size=3pt}]
    \begin{scope}
      \foreach \i/\angle in {1/90,2/30,3/-30,4/-90,5/-150,6/150}
        \node[v,label={\angle:$v_{\i}$}] (g\i) at (\angle:1.3) {};
      \draw (g1)--(g3)--(g5)--(g1)
        (g2)--(g4)--(g6)--(g2)
        (g1)--(g4) (g2)--(g5) (g3)--(g6);
      \node at (0,-2.1) {$G=\overline{C_{6}}$};
    \end{scope}
    \begin{scope}[xshift=5.6cm]
      \foreach \i/\e/\angle in {1/12/90,2/23/30,3/34/-30,4/45/-90,5/56/-150,6/16/150}
        \node[v,label={\angle:$x_{\e}$}] (x\i) at (\angle:2) {};
      \foreach \a/\b in {1/2,2/3,3/4,4/5}
        \draw[dashed,black] (x\a)--(x\b);
      \foreach \a/\b in {1/4,2/5,3/6}
        \draw[dashed,black] (x\a) to[bend right=10] (x\b);
      \draw (x2)--(x6)--(x4);
      \foreach \i/\j/\bend/\where in {
        1/3/0/.6,1/4/10/.65,1/5/0/.5,1/6/0/.5,
        2/4/0/.22,2/5/10/.3,3/5/0/.18,3/6/10/.35,5/6/0/.5}{
        \path (x\i) to[bend left=\bend]
          node[pos=\where,coordinate] (w\i\j) {} (x\j);
        \ifodd\i
          \draw[deletedaux] (x\i) to[bend left={\bend/2}] (w\i\j);
        \else
          \draw[auxiliary] (x\i) to[bend left={\bend/2}] (w\i\j);
        \fi
        \ifodd\j
          \draw[deletedaux] (w\i\j) to[bend left={\bend/2}] (x\j);
        \else
          \draw[auxiliary] (w\i\j) to[bend left={\bend/2}] (x\j);
        \fi
      }
      \foreach \i/\j in {1/3,1/4,1/5,1/6,2/4,2/5,3/5,3/6,5/6}
        \node[smallw] at (w\i\j) {};
      \foreach \i in {1,3,5}{
        \node[deleted] at (x\i) {};
        \draw[black] (x\i) ++(-.1,-.1)--++(.2,.2);
        \draw[black] (x\i) ++(-.1,.1)--++(.2,-.2);
      }
      \node at (0,-2.8) {Full graph $\hat{G}$};
    \end{scope}
    \begin{scope}[xshift=11.25cm]
      \foreach \i/\angle in {1/90,2/30,3/-30,4/-90,5/-150,6/150}
        \node[v,label={\angle:$v_{\i}$}] (g\i) at (\angle:1.3) {};
      \draw[gray] (g1)--(g5)--(g3) (g2)--(g4)--(g6)
        (g1)--(g4) (g2)--(g5);
      \draw[line width=1pt] (g1)--(g3)--(g6)--(g2);
      \draw[added] (g1)--node[above right]{$e_{12}$}(g2)
        (g3)--node[below right]{$e_{34}$}(g4)
        (g5)--node[above left]{$e_{56}$}(g6);
      \foreach \i in {1,2,3,6}
        \node[v,fill=black!20] at (g\i) {};
      \node[align=center] at (0,-2.1) {$G+F$\\with an induced $C_{4}$};
    \end{scope}
  \end{tikzpicture}
  \caption{A minimum OCT whose corresponding nonedges do not form an interval completion.
  In the full graph $\hat{G}$ (middle), small unlabeled auxiliary vertices lie on curved paths, and dashed edges meet $D=\{x_{12},x_{34},x_{56}\}$.
  Three vertex-disjoint triangles give the lower bound of three, while colors $\forwardcolor,\forwardcolor,\reversecolor$ on $x_{23},x_{45},x_{16}$ extend to all auxiliary vertices, so $D$ is minimum.
  In $G+F$ (right), $F=\{e_{12},e_{34},e_{56}\}$ leaves the shaded induced four-cycle $v_{1}v_{2}v_{6}v_{3}v_{1}$ because $e_{16}$ and $e_{23}$ remain absent; yet $F^{*}=\{e_{12},e_{23},e_{34}\}$ is an interval completion of size three.}
  \label{fig:oct-fill-counterexample}
\end{figure}
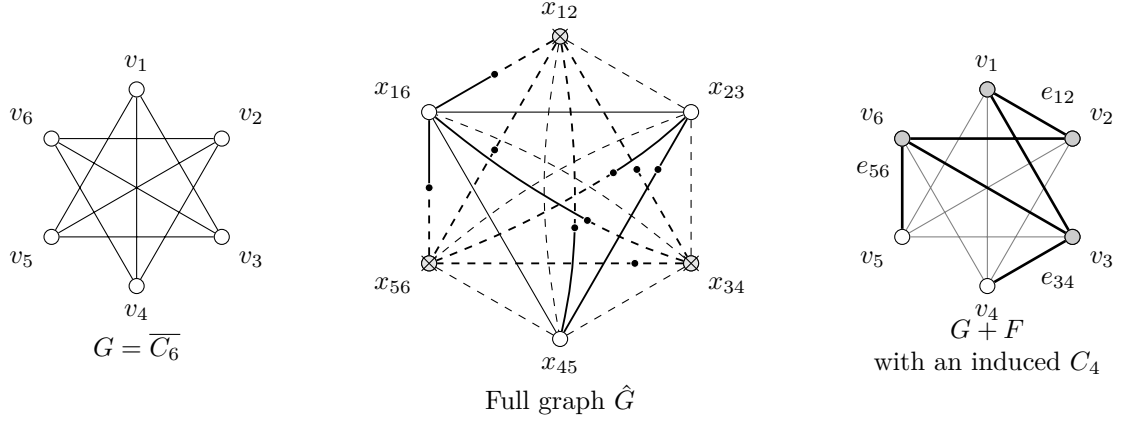

Before proving the reverse inequality $\ic(G)\leq\oct(\hat{G})$, we describe a randomized procedure for constructing an interval representation based on a minimum OCT of $\hat{G}$.
This construction is used only in the analysis; the kernelization in Section~\ref{sec:kernel} uses the instance reductions described there.
Fix a minimum OCT $D$ consisting only of variable vertices and a proper two-coloring $\chi$ of $\hat{G}-D$.
Put $\hat{k}=|D|=\oct(\hat{G})$.
For every retained nonedge $e_{ij}$ with $i<j$, select its separation as follows:
  \begin{equation*}
    \left\{
    \begin{aligned}
      &v_{i}\prec v_{j}\quad&&\text{if }\chi(x_{ij})=\forwardcolor; \\
      &v_{i}\succ v_{j}\quad&&\text{if }\chi(x_{ij})=\reversecolor.
    \end{aligned}
    \right.
  \end{equation*}
  By Lemma~\ref{lem:color-directions}, no two selected separations $u\prec v$ and $x\prec y$ satisfy both $u\sim y$ and $x\sim v$.
  We will use these \emph{selected separations} to construct an interval representation of a completion of $G$.
  Note that these selected separations prescribe initial endpoint comparisons and need not hold in the final interval representation.

  For each $v\in V(G)$, introduce a left endpoint $L_{v}$ and a right endpoint $R_{v}$, so there are $2n$ endpoints in total.
  For pairs of left endpoints satisfying neither of the first two conditions in Table~\ref{tab:endpoint-placement}, sample the initial comparisons uniformly and independently once per unordered pair before recursion.
  These comparisons are never resampled.
  Algorithm~\ref{alg:endpoint-order} recursively orders these endpoints, starting with the block $B_{\mathrm{initial}}=\{L_{v},R_{v}:v\in V(G)\}$.
  \begin{algorithm}[H]
    \caption{Recursive endpoint ordering \textsc{Order}$(B)$.}
    \label{alg:endpoint-order}
    \begin{algorithmic}[1]
      \REQUIRE A block $B$ of endpoints.
      \ENSURE An order of $B$.
      \IF{$B$ contains no left endpoint}
        \RETURN an arbitrary order of $B$.
      \ENDIF
      \STATE Choose $v$ uniformly from $\{u\in V(G):L_{u}\in B\}$.
      \STATE Partition $B\setminus\{L_{v}\}$ into $B_{1}$ and $B_{2}$ using Table~\ref{tab:endpoint-placement}.
      \RETURN the concatenation $\textsc{Order}(B_{1}),L_{v},\textsc{Order}(B_{2})$.
    \end{algorithmic}
  \end{algorithm}
  Assign the coordinates $1,2,\ldots,2n$ to the endpoints $\{ L_{v}, R_{v}:v\in V(G)\}$ according to the returned order, and set $I_{v}=[L_{v},R_{v}]$ for each $v\in V(G)$.

  Table~\ref{tab:endpoint-placement} specifies how Algorithm~\ref{alg:endpoint-order} partitions $B\setminus\{L_{v}\}$ into $B_{1}$ and $B_{2}$ when $L_{v}$ is the pivot.
  \begin{table}[H]
    \centering
    \caption{Endpoint placement with pivot $L_{v}$.}
    \label{tab:endpoint-placement}
    \begin{tabular}{p{0.32\linewidth}p{0.5\linewidth}}
      \toprule
      Placement & Condition \\
      \midrule
      Put $L_{u}$ in $B_{1}$, before $L_{v}$ & There is a selected separation $x\prec v$ with $u\sim x$. \\
      Put $L_{u}$ in $B_{2}$, after $L_{v}$ & There is a selected separation $x\prec u$ with $v\sim x$. \\
      Put $L_{u}$ in $B_{1}$ or $B_{2}$ w.p.\ $1/2$ & Neither preceding condition holds. \\
      \midrule
      Put $R_{u}$ in $B_{1}$, before $L_{v}$ & $u\prec v$ is a selected separation. \\
      Put $R_{u}$ in $B_{2}$, after $L_{v}$ & $u\prec v$ is not a selected separation. \\
      \bottomrule
    \end{tabular}
  \end{table}
  The first two placement rules never compare an endpoint with itself, since a selected separation $x\prec u$ implies $u\not\sim x$.
  The first two conditions for $L_{u}$ cannot both hold, by Lemma~\ref{lem:color-directions}.
  All conditions are evaluated in the full graph $G$ using the fixed selected separations; in particular, $x$ ranges over $V(G)$ even if its endpoints lie outside $B$.
  These rules specify the initial comparison for every pair of distinct endpoints involving a left endpoint; no comparisons between two right endpoints are needed.

  We now prove that this construction establishes the reverse inequality.

\begin{lemma}
  \label{lem:reverse}
  $\ic(G)\leq\oct(\hat{G})$.
\end{lemma}
\begin{proof}
  Run Algorithm~\ref{alg:endpoint-order} on $G$ and the selected separations determined by the minimum OCT $D$ and coloring $\chi$ fixed above.
  The procedure terminates because each pivot is removed from both recursive blocks, and its returned order contains every endpoint exactly once.

  We first show that these intervals represent an interval supergraph $H$ of $G$.
  The last two rows of Table~\ref{tab:endpoint-placement} initially place $L_{u}$ before $R_{x}$ whenever $u\sim x$.
  Initially, no cyclic triangle has comparisons placing $L_{u}$ before $R_{x}$, $R_{x}$ before $L_{v}$, and $L_{v}$ before $L_{u}$ when $u\sim x$.
  Indeed, placing $R_{x}$ before $L_{v}$ means that $x\prec v$ is a selected separation, and $u\sim x$ then requires $L_{u}$ to precede $L_{v}$.
  A pivot fixes all comparisons across its blocks in the order front, pivot, back, while leaving comparisons within each block unchanged.
  In particular, comparisons within a recursive block still agree with the initial table, so applying the same table in subsequent calls uses precisely the current comparisons in that block.
  Any cyclic triangle after the split must therefore lie entirely within $B_{1}$ or $B_{2}$, where all comparisons are unchanged.
  Thus no pivot creates a new cyclic triangle of two left endpoints and one right endpoint.
  Reversing a comparison that places $L_{u}$ before $R_{x}$ would require $R_{x}$ in the front block and $L_{u}$ in the back block, giving exactly such a cyclic triangle before the split.
  The absence of such triangles with $u\sim x$ is preserved throughout the recursion, so every comparison placing $L_{u}$ before $R_{x}$ with $u\sim x$ is preserved.
  In particular, $L_{v}$ precedes $R_{v}$ for every $v$, so each interval is valid.
  For every original edge $uv$, both $L_{u}$ precedes $R_{v}$ and $L_{v}$ precedes $R_{u}$, so $I_{u}$ and $I_{v}$ intersect.
  Thus $H$ is an interval supergraph of $G$.
  Let $F=E(H)\setminus E(G)$ be its added edge set, which is an interval completion of $G$.

  We condition on the random comparisons between left endpoints and show that $\mathbb E[|F|]=\hat{k}$ under the uniform left-pivot choices.
  Let $\sigma$ count left--right pairs whose current comparison puts the right endpoint first.
  Initially each retained nonedge contributes exactly one such comparison, so
  \begin{equation*}
    \sigma_{\mathrm{initial}}=\bar{m}-\hat{k}.
  \end{equation*}

  For a call whose block contains $\ell>0$ left endpoints, sum the change in $\sigma$ over its $\ell$ possible left pivots.
  Only reversed left--right comparisons affect $\sigma$.
  For a fixed pivot, every reversed left--right comparison arises from exactly one cyclic triangle containing the pivot, another left endpoint, and one right endpoint.
  For each such cyclic triangle, choosing either left endpoint as pivot reverses one of its left--right comparisons.
  Of the two left--right comparisons in the triangle, one puts the right endpoint first and the other puts it last.
  The two pivot choices therefore change $\sigma$ by $+1$ and $-1$, so their contributions cancel.
  Triangles containing only left endpoints contribute zero, and comparisons between right endpoints do not affect $\sigma$.
  Thus the expected change in this call is zero.
  A block containing only right endpoints changes no left--right comparison and leaves $\sigma$ unchanged.
  Induction over the recursive blocks gives $\mathbb E[\sigma_{\mathrm{final}}]=\sigma_{\mathrm{initial}}$.
  In the resulting representation, each pair of disjoint intervals contributes exactly one comparison with the right endpoint first.
  Hence
  \begin{equation*}
    \sigma_{\mathrm{final}}=\bar{m}-|F|,
    \qquad\mathbb E[|F|]=\hat{k}.
  \end{equation*}

  Every output defines an interval completion, so Lemma~\ref{lem:forward} gives
  \begin{equation*}
    |F|\geq\ic(G)\geq\oct(\hat{G})=\hat{k}.
  \end{equation*}
  Since $\mathbb E[|F|]=\hat{k}$ and every choice of the initially sampled comparisons and left pivots has positive probability, every output satisfies $|F|=\hat{k}$.
  In particular, $\ic(G)\leq\oct(\hat{G})$, as required.
  Thus both the initially sampled comparisons and the left-pivot choices can be fixed arbitrarily.
  The resulting deterministic procedure produces a minimum interval completion of the same size as $D$.
  The construction takes polynomial time once $D$ and its coloring are given.
  Together with Lemma~\ref{lem:forward}, this gives
  \begin{equation*}
    \ic(G)=\oct(\hat{G}).
  \end{equation*}
  The construction in the proof of Lemma~\ref{lem:forward} then shows that $D^{*}=\{x_{e}:e\in F\}$ is a minimum OCT whose corresponding nonedges form this completion.
\end{proof}

For the original instance $(G,k)$ and the graph $\hat{G}$ constructed in this section, Lemmas~\ref{lem:forward} and~\ref{lem:reverse} give $\ic(G)=\oct(\hat{G})$.
Since $\hat{G}$ can be constructed in $O(n^{4})$ time, the map $(G,k)\mapsto(\hat{G},k)$ is a parameter-preserving polynomial-time reduction, proving Theorem~\ref{thm:reduction}.

\section{{Polynomial kernels}}
\label{sec:kernel}
We apply an OCT kernel to the resulting instance $(\hat{G},k)$ obtained in Section~\ref{sec:construction}
and let $(G^{\prime},k^{\prime})$ be its output.
Next, we give a polynomial-time reduction from the instance $(G^{\prime},k^{\prime})$ to an equivalent \textsc{Interval Completion} instance $(G^{\prime\prime},k^{\prime\prime})$.

\subsection{Reduction from OCT to interval completion}
\label{sec:direct-oct-ic}

Given an OCT instance $(G^{\prime},k^{\prime})$, let $n^{\prime}=|V(G^{\prime})|$ and $m^{\prime}=|E(G^{\prime})|$.
If $k^{\prime}<0$ or $k^{\prime}\geq n^{\prime}$, we can directly return a trivial No-instance or Yes-instance, respectively.
Now, we assume $0\leq k^{\prime} < n^{\prime}$.

First, we construct an auxiliary graph $H$ from two disjoint copies of $G^{\prime}$, retaining all edges within each copy and adding an edge between the two copies of every vertex.
Formally, for each $v\in V(G^{\prime})$, we introduce vertices $v^{0}$ and $v^{1}$ and add the edge $v^{0}v^{1}$.
For each $uv\in E(G^{\prime})$ and $i=0,1$, we also add the edge $u^{i}v^{i}$.
Thus the induced subgraph of $H$ on $\{v^{i}:v\in V(G^{\prime})\}$ is isomorphic to $G^{\prime}$ for each $i=0,1$.
Clearly, $H$ has $2n^{\prime}$ vertices and $n^{\prime}+2m^{\prime}$ edges, i.e.,
\begin{equation*}
  n_{H}=|V(H)|=2n^{\prime}; \qquad
  m_{H}=|E(H)|=n^{\prime}+2m^{\prime}.
\end{equation*}
Let $k_{H}=n^{\prime}-k^{\prime}$.
Then $H$ has an independent set of size at least $k_{H}$ if and only if $G^{\prime}$ has an OCT of size at most $k^{\prime}$.
Indeed, if we regard $i \in \{0, 1\}$ as the color of $v^{i}$, an independent set selects at most one copy of each vertex and its copy labels properly color the corresponding vertices of $G^{\prime}$; conversely, a two-coloring after deleting an OCT gives an independent set by selecting the copy matching each remaining vertex's color.

Now, we construct an instance $(G^{\prime\prime},k^{\prime\prime})$ of \textsc{Interval Completion} such that $H$ is a subgraph of $G^{\prime\prime}$.
The graph $G^{\prime\prime}$ has vertex set $P\mathbin{\dot\cup}Q$, constructed as follows.
\begin{enumerate}[label=\arabic*.,leftmargin=*]
  \item Introduce $r=m_{H}n_{H}$ \emph{redundant vertices} and denote their set by $R$, and let $P=V(H)\mathbin{\dot\cup}R$.
  \item Introduce a set $Q$ of $\binom{n_{H}}{2}-m_{H}+4m^{\prime}+rn_{H}$ new vertices with the following three types of neighborhoods in $P$.
  \begin{enumerate}[label=(\roman*),leftmargin=*]
    \item For each nonedge $uv \in \bar{E}(H)$, add a vertex $y^{uv}$ with neighborhood $\{u,v\}$.
    \item For each $v\in V(H)$, add vertices $y^{v}_{j}$ for $1\leq j\leq d_{H}(v)-1$, each with neighborhood $\{v\}$.
    \item For each $v\in V(H)$, add $r$ vertices $z^{v}_{1},z^{v}_{2},\ldots,z^{v}_{r}$, each with neighborhood $\{v\}\cup R$.
  \end{enumerate}
  In the labels of vertices in $Q$, superscripts identify the associated nonedge or vertex of $H$, and subscripts index vertices with the same neighborhood.
  \item Add edges to make $P$ and $Q$ cliques, obtaining $G^{\prime\prime}$.
\end{enumerate}
There are $4m^{\prime}$ vertices of type~(ii) in total.

Let $n^{\prime\prime}=|V(G^{\prime\prime})|$ and $m^{\prime\prime}=|E(G^{\prime\prime})|$.
Define
\begin{equation*}
    k^{\prime\prime}=C_{0}+m_{H}(n_{H}-1)-r\cdot k_{H},
\end{equation*}
where $C_{0}=\frac{1}{6}n_{H}(n_{H}-1)(2n_{H}-7)+rn_{H}(n_{H}-1)$.
The output is $(G^{\prime\prime},k^{\prime\prime})$.
Figure~\ref{fig:oct-ic-construction} illustrates the three graphs in this construction.

\begin{figure}[H]
  \centering
  \begin{tikzpicture}[
      font=\small,
      v/.style={circle,draw=black,fill=white,minimum size=4pt,inner sep=0pt},
      part/.style={draw=black!65,rounded corners=3pt,fill=black!2},
      group/.style={draw=black!45,rounded corners=2pt,fill=white},
      flow/.style={-{Stealth[length=2mm]},line width=.8pt},
      every label/.style={font=\scriptsize,text=black}
    ]
    \node at (1.8,3.15) {$G^{\prime}=K_{3}$};
    \node at (1.8,.15) {$H$};
    \node at (9.675,3.15) {$G^{\prime\prime}$};

    \node[v,label=left:$v_{1}$] (g1) at (.6,2.1) {};
    \node[v,label=above:$v_{2}$] (g2) at (1.8,2.1) {};
    \node[v,label=right:$v_{3}$] (g3) at (3,2.1) {};
    \draw (g1)--(g2)--(g3);
    \draw (g1) .. controls (.6,3) and (3,3) .. (g3);
    \draw[flow] (1.8,1.8)--(1.8,.5);

    \node[v,label={[yshift=2pt]left:$v_{1}^{0}$}] (h10) at (.6,-.8) {};
    \node[v,label={[yshift=2pt]above:$v_{2}^{0}$}] (h20) at (1.8,-.8) {};
    \node[v,label={[yshift=2pt]right:$v_{3}^{0}$}] (h30) at (3,-.8) {};
    \node[v,label={[yshift=2pt]left:$v_{1}^{1}$}] (h11) at (.6,-1.95) {};
    \node[v,label={[yshift=2pt]below:$v_{2}^{1}$}] (h21) at (1.8,-1.95) {};
    \node[v,label={[yshift=2pt]right:$v_{3}^{1}$}] (h31) at (3,-1.95) {};
    \draw (h10)--(h20)--(h30);
    \draw (h10) .. controls (.6,.1) and (3,.1) .. (h30);
    \draw (h11)--(h21)--(h31);
    \draw (h11) .. controls (.6,-2.85) and (3,-2.85) .. (h31);
    \draw (h10)--(h11) (h20)--(h21) (h30)--(h31);
    \draw[flow] (3.65,-1.375)--(4.2,-1.375);

    \draw[part] (4.85,.35) rectangle (14.5,2.775);
    \draw[part] (4.85,-2.625) rectangle (14.5,-.3);
    \node[anchor=west] at (5.05,2.53) {$P$ (clique)};
    \node[anchor=west] at (5.05,-2.4) {$Q$ (clique)};
    \draw[group] (5.1,.55) rectangle (9.55,2.3);
    \node[font=\scriptsize] at (7.325,2.12) {$V(H)$};
    \draw[group] (9.8,.55) rectangle (14.25,2.3);
    \node[font=\scriptsize] at (12.025,2.12) {$R$ ($54$ vertices)};
    \coordinate (rgroup) at (12.025,.55);
    \draw[group] (7.825,-2.25) rectangle (8.975,-.8);
    \draw[group] (9.625,-2.25) rectangle (10.775,-.8);
    \draw[group] (11.225,-2.25) rectangle (12.375,-.8);
    \draw[group] (13.075,-2.25) rectangle (14.225,-.8);
    \coordinate (qy) at (8.4,-.8);
    \coordinate (quy) at (10.2,-.8);
    \coordinate (qz) at (11.8,-.8);
    \coordinate (quz) at (13.65,-.8);

    \node[v,label=right:$v_{1}^{0}$] (p10) at (5.8,1.4) {};
    \node[v,label=right:$v_{2}^{0}$] (p20) at (7.3,1.4) {};
    \node[v,label=right:$v_{3}^{0}$] (px) at (8.8,1.4) {};
    \node[v,label=right:$v_{1}^{1}$] (p11) at (5.8,.95) {};
    \node[v,label=right:$v_{2}^{1}$] (py) at (7.3,.95) {};
    \node[v,label=right:$v_{3}^{1}$] (p31) at (8.8,.95) {};
    \node[v,label=right:$x_{1}$] (p1) at (10.4,1.425) {};
    \node[v,label=right:$x_{2}$] (p2) at (11.35,1.425) {};
    \node[font=\scriptsize] at (12.35,1.425) {$\cdots$};
    \node[v,label=right:$x_{54}$] (p54) at (13.4,1.425) {};

    \node[v,label=below:$y^{v_{1}^{0}v_{2}^{1}}$] (q1a) at (5.4,-1.2) {};
    \node[v,label=below:$y^{v_{3}^{0}v_{2}^{1}}$] (q1b) at (7,-1.2) {};
    \node[font=\scriptsize] at ($(q1a.center)!0.5!(q1b.center)$) {$\cdots$};
    \node[v,label=right:$y^{v_{1}^{0}}_{1}$] (q2a) at (8.025,-1.2) {};
    \node[v,label=right:$y^{v_{1}^{0}}_{2}$] (q2b) at (8.025,-1.9) {};
    \node[font=\scriptsize] at (9.3,-1.55) {$\cdots$};
    \node[v,label=right:$y^{v_{3}^{1}}_{1}$] (qu2a) at (9.825,-1.2) {};
    \node[v,label=right:$y^{v_{3}^{1}}_{2}$] (qu2b) at (9.825,-1.9) {};
    \node[v,label=right:$z^{v_{1}^{0}}_{1}$] (q3a) at (11.425,-1.15) {};
    \node[v,label=right:$z^{v_{1}^{0}}_{54}$] (q3c) at (11.425,-1.95) {};
    \node[font=\scriptsize,scale=.7,yshift=2pt] at ($(q3a.center)!0.5!(q3c.center)$) {$\vdots$};
    \node[font=\scriptsize] at (12.75,-1.55) {$\cdots$};
    \node[v,label=right:$z^{v_{3}^{1}}_{1}$] (qu3a) at (13.275,-1.15) {};
    \node[v,label=right:$z^{v_{3}^{1}}_{54}$] (qu3c) at (13.275,-1.95) {};
    \node[font=\scriptsize,scale=.7,yshift=2pt] at ($(qu3a.center)!0.5!(qu3c.center)$) {$\vdots$};

    \draw (p10)--(q1a) (py)--(q1a);
    \draw (px)--(q1b) (py)--(q1b);
    \draw (p10)--(qy);
    \draw (p31)--(quy);
    \draw (p10)--(qz);
    \draw (p31)--(quz);
    \draw (rgroup)--(qz) (rgroup)--(quz);
  \end{tikzpicture}
  \caption{The construction for $G^{\prime}=K_{3}$, with superscripts $0$ and $1$ distinguishing the two copies of each vertex of $G^{\prime}$ in $H$ and $P$.
  Here $R=\{x_{1},x_{2},\ldots,x_{54}\}$; $Q$ shows selected nonedge vertices and boxed groups associated with $v_{1}^{0}$ and $v_{3}^{1}$.
  Lines from a vertex to a box represent adjacency to every vertex inside; a line between boxes represents all edges between their vertex sets.
  Both $P$ and $Q$ are cliques, with all internal edges omitted.}
  \label{fig:oct-ic-construction}
\end{figure}
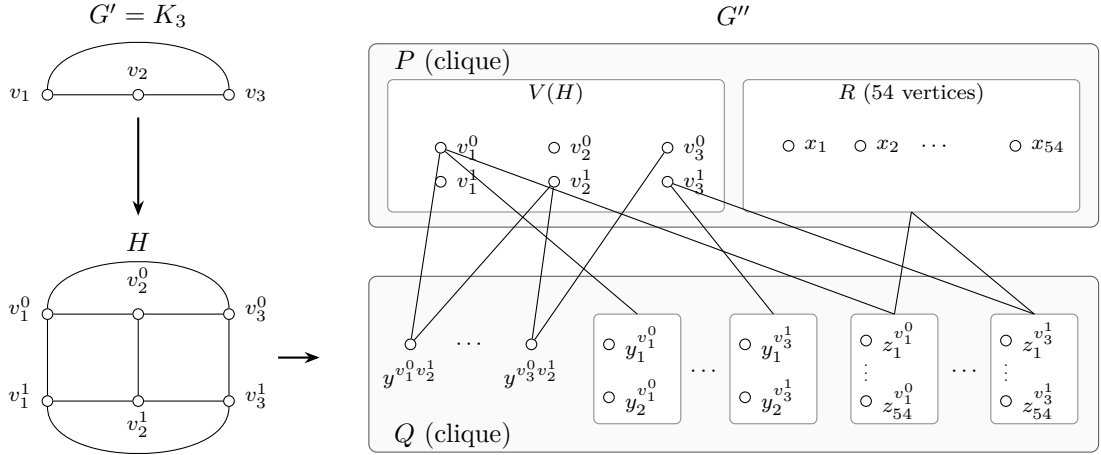

We now prove that the constructed \textsc{Interval Completion} instance is equivalent to the original OCT instance.
The main idea is to show that $H$ has an independent set of size at least $k_{H}$ if and only if $G^{\prime\prime}$ has an interval completion of size at most $k^{\prime\prime}$.
The introduction of redundant vertices amplifies the effect of the independent set size on the completion cost, making it dominate other variations caused by the vertex ordering.

\begin{lemma}
  \label{lem:direct-oct-ic}
  The OCT instance $(G^{\prime},k^{\prime})$ is a Yes-instance if and only if the constructed \textsc{Interval Completion} instance $(G^{\prime\prime},k^{\prime\prime})$ is a Yes-instance.
  Moreover, it satisfies $n^{\prime\prime}=O(N^{3})$, $m^{\prime\prime}=O(N^{6})$, and $k^{\prime\prime}=O(N^{4})$, where $N=n^{\prime}+m^{\prime}$.
\end{lemma}
\begin{proof}
  Since $P$ and $Q$ are cliques, it is easy to show by induction that the intervals of each clique have a common point in any interval representation of a completion.
  Reflecting the line if necessary, assume that the common point of $Q$ is no further right than that of $P$.
  For $x\in P$ and $y\in Q$, we then have $L_{y}\leq R_{x}$, so $xy$ is an edge exactly when $L_{x}\leq R_{y}$.
  We may therefore extend all right endpoints of $P$ to a common point beyond every original right endpoint and all left endpoints of $Q$ to another common point before every original left endpoint, without changing the represented graph.
  Break ties among the remaining endpoints by placing left endpoints before right endpoints, preserving intersections of closed intervals.
  Choose coordinates so that the common right endpoint of $P$ is $|P|+1$, the common left endpoint of $Q$ is $0$, and the left endpoints of $P$ occupy positions $1,2,\ldots,|P|$.
  Let $\pi_{P}(x)$ record the left-endpoint position of $x\in P$, and let $\pi_{Q}(y)$ record the right-endpoint position of $y\in Q$ on the same axis.
  Thus
  \begin{equation*}
    I_{x}=[\pi_{P}(x),|P|+1]\quad(x\in P),\qquad
    I_{y}=[0,\pi_{Q}(y)]\quad(y\in Q).
  \end{equation*}
  Figure~\ref{fig:pq-endpoints} illustrates these endpoint positions.
  \begin{figure}[H]
    \centering
    \begin{tikzpicture}[
        x=1cm,y=1cm,font=\small,
        guide/.style={densely dashed,black!45,line width=.4pt},
        sample/.style={black!45,line width=.7pt},
        pinterval/.style={black,line width=1.1pt},
        qinterval/.style={black,line width=1.1pt}
      ]
      \draw[guide] (0,0)--(0,1.85);
      \draw[guide] (11,0)--(11,-1.85);
      \draw[guide] (4,-1.2)--(4,.05);
      \draw[guide] (6.5,0)--(6.5,1.2);
      \draw[line width=.6pt] (-.15,0)--(11.2,0);
      \foreach \t in {0,1,2,3,4,5,6,7,8,9,10,11} {
        \draw (\t,-.06)--(\t,.06);
      }
      \node[below=3pt] at (0,0) {$0$};
      \node[above=3pt] at (4,0) {$\pi_{P}(x)$};
      \node[below=3pt] at (6.5,0) {$\pi_{Q}(y)$};
      \node[above=3pt] at (11,0) {$|P|+1$};
      \node[anchor=east] at (-.35,1.2) {$Q$};
      \node[anchor=east] at (-.35,-1.2) {$P$};

      \draw[sample] (0,1.7)--(2.5,1.7);
      \draw[qinterval] (0,1.2)--(6.5,1.2);
      \draw[sample] (0,.7)--(8.5,.7);
      \foreach \a/\b in {0/1.7,2.5/1.7,0/.7,8.5/.7} {
        \fill[black!45] (\a,\b) circle (1.2pt);
      }
      \fill[black] (0,1.2) circle (1.4pt);
      \fill[black] (6.5,1.2) circle (1.4pt);
      \node[above=2pt,text=black] at (5.3,1.2) {$I_{y}$};

      \draw[sample] (2,-.7)--(11,-.7);
      \draw[pinterval] (4,-1.2)--(11,-1.2);
      \draw[sample] (8,-1.7)--(11,-1.7);
      \foreach \a/\b in {2/-.7,11/-.7,8/-1.7,11/-1.7} {
        \fill[black!45] (\a,\b) circle (1.2pt);
      }
      \fill[black] (4,-1.2) circle (1.4pt);
      \fill[black] (11,-1.2) circle (1.4pt);
      \node[below=2pt,text=black] at (5.3,-1.2) {$I_{x}$};
    \end{tikzpicture}
    \caption{Representative intervals after endpoint normalization.
    The intervals of $Q$ share left endpoint $0$, and the intervals of $P$ share right endpoint $|P|+1$.
    Right endpoints of $Q$ are placed half a unit after integer positions.
    In the represented completion, $xy$ is an edge exactly when $\pi_{P}(x)\leq\pi_{Q}(y)$, as illustrated by the thicker intervals.}
    \label{fig:pq-endpoints}
  \end{figure}
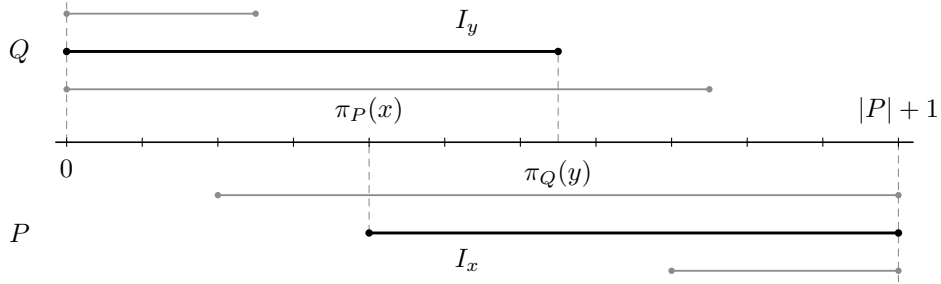
  Each vertex of $Q$ consequently has a prefix of $\pi_{P}$ as its neighborhood in $P$.
  For $y\in Q$, write $N_{P}(y)$ for its original neighborhood in $P$.
  Once $\pi_{P}$ is fixed, choosing $\pi_{Q}(y)=\max_{x\in N_{P}(y)}\pi_{P}(x)+\tfrac12$ minimizes the number of added edges incident with $y$ while retaining all its original neighbors.
  This number is the position of its last original neighbor minus $|N_{P}(y)|$.
  Conversely, every permutation $\pi_{P}$ with these right-endpoint positions gives an interval completion.
  Hence
  \begin{equation*}
    \ic(G^{\prime\prime})
    =\min_{\pi_{P}}\sum_{y\in Q}
    \bigl(\max_{x\in N_{P}(y)}\pi_{P}(x)-|N_{P}(y)|\bigr).
  \end{equation*}

  Let $\alpha(H)$ denote the maximum size of an independent set of $H$.
  The correspondence between independent sets of $H$ and two-colorings of induced subgraphs of $G^{\prime}$ gives~\cite{FellowsEtAl2012}
  \begin{equation*}
    \alpha(H)=n^{\prime}-\oct(G^{\prime}).
  \end{equation*}
  It suffices to prove that $\ic(G^{\prime\prime})\leq k^{\prime\prime}$ if and only if $\alpha(H)\geq k_{H}$.

  For a fixed permutation $\pi_{P}$, the contribution of each vertex of $Q$ is given by its type:
  \begin{equation*}
    \left\{
    \begin{aligned}
      &C(y^{uv}) = \max\{\pi_{P}(u),\pi_{P}(v)\}-2;\\
      &C(y^{v}_{j}) = \pi_{P}(v)-1;\\
      &C(z^{v}_{j}) = \max\bigl\{\pi_{P}(v),\max_{x\in R}\pi_{P}(x)\bigr\}-r-1.
    \end{aligned}
    \right.
  \end{equation*}
  So we have $\min_{\pi_{P}} \sum_{y \in Q} C(y) = \ic(G^{\prime\prime})$.
  We next show that a minimum is attained by an order in which $R$ is consecutive.

  Let $\pi_{P}$ be an order of the vertices in $P$, and let $x$ be the last vertex of $R$ in $\pi_{P}$.
  Consider updating the order $\pi_{P}$ by moving all other vertices of $R$ immediately before $x$, preserving the relative order of $V(H)$.
  Every vertex of $H$ moves only to an earlier position or stays in place.
  The contributions from types~(i) and~(ii) cannot increase.
  The last vertex $x$ of $R$ and all vertices of $H$ originally after it keep their positions, so the contributions from type~(iii) are unchanged.
  This operation therefore cannot increase the number of added edges, so a minimum is attained by an order in which $R$ is consecutive.

  Now we know that $R$ is consecutive in the order $\pi_{P}$ satisfying $\ic(G^{\prime\prime})=\sum_{y\in Q}C(y)$.
  Let $\pi$ denote the relative order of $V(H)$ in $\pi_{P}$, with positions $1,2,\ldots,n_{H}$, and let $I$ be the set of vertices after $R$.
  When $R$ is at the end, we have $I=\varnothing$ and $\pi_{P}(v)=\pi(v)$ for every $v\in V(H)$.
  Each vertex of type~(iii) then contributes $(n_{H}+r)-(r+1)=n_{H}-1$, so their total contribution is $r\cdot n_{H}(n_{H}-1)$.
  For types~(i) and~(ii), write the sum over nonedges as the sum over all unordered vertex pairs minus the sum over edges.
  The degree terms cancel using
  \begin{equation*}
    \max\{\pi(u),\pi(v)\}=\pi(u)+\pi(v)-\min\{\pi(u),\pi(v)\}.
  \end{equation*}
  The combined contribution of these two types is
  \begin{equation*}
    \begin{aligned}
      &\sum_{uv\in\bar{E}(H)}\bigl(\max\{\pi(u),\pi(v)\}-2\bigr)
      +\sum_{v\in V(H)}(d_{H}(v)-1)(\pi(v)-1)\\
      ={}&\sum_{j=1}^{n_{H}}(j-1)(j-2)-\sum_{j=1}^{n_{H}}(j-1)
      +\sum_{uv\in E(H)}\min\{\pi(u),\pi(v)\}\\
      ={}&\frac{1}{6}n_{H}(n_{H}-1)(2n_{H}-7)
      +\sum_{uv\in E(H)}\min\{\pi(u),\pi(v)\}.
    \end{aligned}
  \end{equation*}
  Adding the contribution of type~(iii) and applying the definition of $C_{0}$ gives
  \begin{equation*}
    C_{0}+\sum_{uv\in E(H)}\min\{\pi(u),\pi(v)\}.
  \end{equation*}
  Moving the block $R$ left across a vertex decreases the total contribution by $r$, while each edge from that vertex to a vertex already after $R$ increases it by $r$.
  Therefore, for any order in which $R$ is consecutive, we have
  \begin{equation*}
    \sum_{y\in Q}C(y)=C_{0}+\sum_{uv\in E(H)}\min\{\pi(u),\pi(v)\}
    -r\cdot\bigl(|I|-|E(H[I])|\bigr).
  \end{equation*}
  Each term in the sum lies between $1$ and $n_{H}-1$.
  Moreover, we observe that $|I|-|E(H[I])|\leq\alpha(H)$, this is because removing a vertex with a neighbor in $I$ cannot decrease this quantity, and repeating this operation produces an independent set.
  Since a minimum is attained by an order in which $R$ is consecutive, we obtain
  \begin{equation*}
    \ic(G^{\prime\prime})\geq C_{0}+m_{H}-r\cdot\alpha(H).
  \end{equation*}

  For the upper bound, take a maximum independent set $I$ of $H$, put it after $R$, and put all other vertices of $H$ before $R$.
  Since $E(H[I])=\varnothing$, we obtain
  \begin{equation*}
    \ic(G^{\prime\prime})\leq C_{0}+m_{H}(n_{H}-1)-r\cdot\alpha(H).
  \end{equation*}

  If $\alpha(H)\geq k_{H}$, the upper bound gives $\ic(G^{\prime\prime})\leq k^{\prime\prime}$.
  If $\alpha(H)\leq k_{H}-1$, the lower bound gives
  \begin{equation*}
      \ic(G^{\prime\prime}) \geq C_{0}+m_{H}-r\cdot(k_{H}-1)
      =k^{\prime\prime}+r-m_{H}(n_{H}-2)
      =k^{\prime\prime}+2m_{H}>k^{\prime\prime},
  \end{equation*}
  where the last equality uses $r=m_{H}n_{H}$, and $m_{H}\geq n^{\prime}\geq1$.
  This proves the equivalence.

  Finally, counting the three types of vertices in $Q$ gives
  \begin{equation*}
      |Q|=\binom{n_{H}}{2}-m_{H}+4m^{\prime}+rn_{H}; \quad
      n^{\prime\prime}=n_{H}+r+|Q|.
  \end{equation*}
  Since $r=O(N^{2})$ and $n_{H}=O(N)$, we obtain $n^{\prime\prime}=O(N^{3})$ and $m^{\prime\prime}=O(N^{6})$.
  The parameter is an integer because the numerator in the first term of $C_{0}$ is divisible by six.
  It is positive since $n_{H}\geq2$, $k_{H}\leq n_{H}/2$, $m_{H}\geq1$, and $r\geq2$.
  Its defining formula gives $k^{\prime\prime}=O(N^{4})$.
  All vertices and edges can be constructed in deterministic polynomial time.
\end{proof}

\subsection{Transferring OCT kernels}
\label{sec:kernel-transfer}

\begin{lemma}
  \label{lem:oct-kernel-transfer}
  {An OCT kernel with at most $g(k)$ vertices and edges in total yields an \textsc{Interval Completion} kernel with $O(g(k)^{3})$ vertices, $O(g(k)^{6})$ edges, and output parameter bounded by $O(g(k)^{4})$.}
\end{lemma}
\begin{proof}

  If $k=0$, recognize interval graphs and return the fixed Yes-instance $(K_{1},0)$ or No-instance $(C_{4},0)$ accordingly.
  For $k\geq1$, combine Theorem~\ref{thm:reduction} with the OCT kernel and let $(G^{\prime},k^{\prime})$ be its output, where $N=n^{\prime}+m^{\prime}\leq g(k)$.
  If $k^{\prime}<0$, return $(C_{4},0)$; if $k^{\prime}\geq n^{\prime}$, return $(K_{1},0)$.

  Lemma~\ref{lem:direct-oct-ic} gives an equivalent interval completion instance $(G^{\prime\prime},k^{\prime\prime})$ with $n^{\prime\prime}=O(N^{3})=O(g(k)^{3})$, $m^{\prime\prime}=O(N^{6})=O(g(k)^{6})$, and $k^{\prime\prime}=O(N^{4})=O(g(k)^{4})$.
  After relabeling, the adjacency matrix and the encoded vertex count and parameter use $O(g(k)^{6})$ bits in total.
  The return reduction takes deterministic polynomial time and preserves equivalence, so the resulting kernel inherits both the error probability and the error direction of the OCT kernel.
\end{proof}

\begin{proof}[Proof of Theorems~\ref{thm:kernel} and~\ref{thm:kernel-1}]
  The deterministic OCT kernel of Koana and Kumabe~\cite{KoanaKumabe2026} has $g(k)=O(k^{12})$.

  For the randomized bound, apply the kernel of Kratsch and Wahlstr\"om~\cite{KratschWahlstrom2020} after adding $n$ isolated vertices to $\hat G$.
  This preserves the answer and the parameter $k$.
  The output has $O((k\log k\log\log k)^{3})$ vertices whose edges prescribe equal or different colors.
  Its failure probability is at most $2^{-\Omega(n)}$, with errors limited to false negatives.
  Replace each equality constraint by a path of length two, and keep each inequality constraint as an edge.
  A deleted internal vertex of such a path can be replaced by an endpoint without increasing the number of deletions.
  This gives an equivalent ordinary OCT instance with $O((k\log k\log\log k)^{6})=\widetilde O(k^{6})$ vertices and edges in total.
  This conversion and the return reduction preserve the error guarantee.

  Applying Lemma~\ref{lem:oct-kernel-transfer} with $g(k)=O(k^{12})$ proves Theorem~\ref{thm:kernel}.
  With $g(k)=\widetilde O(k^{6})$, it proves Theorem~\ref{thm:kernel-1}.
\end{proof}

\section{Conclusion}
\label{sec:conclusion}
{We obtain the first polynomial kernel for \textsc{Interval Completion} through a parameter-preserving reduction to OCT.
The construction represents each nonedge of $G$ by a variable vertex in the OCT instance, 
and uses edges and length-two paths between vertices to encode constraints from $G$.
Filling a nonedge in $G$ then corresponds to deleting a vertex in the OCT instance.}

Although we provide polynomial kernels for Interval Completion, whether polynomial kernels exist for the Interval Edge Deletion and Interval Edge Editing problems remains unknown. Interval Edge Editing allows both edge additions and deletions. These are all worthy of further study.


\begingroup
\section{AI statement}
\label{sec:ai-statement}
The initial version of the main proof was developed through discussions between the authors and the AI model GPT-6 Astra.
We simplified it and rewrote it based on our own understanding.
We also used this model to polish the text and create schematic figures.
The authors take full responsibility for the content and correctness of the paper.
\par\endgroup


\begingroup
\small
\bibliographystyle{unsrt}
\bibliography{references}

@article{BoothLueker1976,
  author  = {Kellogg S. Booth and George S. Lueker},
  title   = {Testing for the consecutive ones property, interval graphs, and graph planarity using {PQ}-tree algorithms},
  journal = {Journal of Computer and System Sciences},
  volume  = {13},
  number  = {3},
  pages   = {335--379},
  year    = {1976},
  doi     = {10.1016/S0022-0000(76)80045-1},
  url     = {https://doi.org/10.1016/S0022-0000(76)80045-1},
  note    = {\url{https://doi.org/10.1016/S0022-0000(76)80045-1}}
}

@article{VillangerEtAl2009,
  author  = {Yngve Villanger and Pinar Heggernes and Christophe Paul and Jan Arne Telle},
  title   = {Interval completion is fixed parameter tractable},
  journal = {SIAM Journal on Computing},
  volume  = {38},
  number  = {5},
  pages   = {2007--2020},
  year    = {2009},
  doi     = {10.1137/070710913},
  url     = {https://doi.org/10.1137/070710913},
  note    = {\url{https://doi.org/10.1137/070710913}}
}

@article{KaplanShamirTarjan1999,
  author  = {Haim Kaplan and Ron Shamir and Robert E. Tarjan},
  title   = {Tractability of parameterized completion problems on chordal, strongly chordal, and proper interval graphs},
  journal = {SIAM Journal on Computing},
  volume  = {28},
  number  = {5},
  pages   = {1906--1922},
  year    = {1999},
  doi     = {10.1137/S0097539796303044},
  url     = {https://doi.org/10.1137/S0097539796303044},
  note    = {\url{https://doi.org/10.1137/S0097539796303044}}
}

@article{FominGolovach2000,
  author  = {Fedor V. Fomin and Petr A. Golovach},
  title   = {Graph searching and interval completion},
  journal = {SIAM Journal on Discrete Mathematics},
  volume  = {13},
  number  = {4},
  pages   = {454--464},
  year    = {2000},
  doi     = {10.1137/S0895480199350477},
  url     = {https://doi.org/10.1137/S0895480199350477},
  note    = {\url{https://doi.org/10.1137/S0895480199350477}}
}

@article{BessyPerez2013,
  author  = {St{\'e}phane Bessy and Anthony Perez},
  title   = {Polynomial kernels for {Proper Interval Completion} and related problems},
  journal = {Information and Computation},
  volume  = {231},
  pages   = {89--108},
  year    = {2013},
  doi     = {10.1016/j.ic.2013.08.006},
  url     = {https://doi.org/10.1016/j.ic.2013.08.006},
  note    = {\url{https://doi.org/10.1016/j.ic.2013.08.006}}
}

@article{BliznetsEtAl2018,
  author  = {Ivan Bliznets and Fedor V. Fomin and Marcin Pilipczuk and Micha{\l} Pilipczuk},
  title   = {Subexponential parameterized algorithm for interval completion},
  journal = {ACM Transactions on Algorithms},
  volume  = {14},
  number  = {3},
  pages   = {35:1--35:62},
  year    = {2018},
  doi     = {10.1145/3186896},
  url     = {https://doi.org/10.1145/3186896},
  note    = {Article 35. \url{https://doi.org/10.1145/3186896}}
}

@article{CrespelleEtAl2023,
  author  = {Christophe Crespelle and P{\aa}l Gr{\o}n{\aa}s Drange and Fedor V. Fomin and Petr A. Golovach},
  title   = {A survey of parameterized algorithms and the complexity of edge modification},
  journal = {Computer Science Review},
  volume  = {48},
  pages   = {100556},
  year    = {2023},
  doi     = {10.1016/j.cosrev.2023.100556},
  url     = {https://doi.org/10.1016/j.cosrev.2023.100556},
  note    = {\url{https://doi.org/10.1016/j.cosrev.2023.100556}}
}

@inproceedings{Cao2016,
  author    = {Yixin Cao},
  title     = {Linear recognition of almost interval graphs},
  booktitle = {Proceedings of the Twenty-Seventh Annual {ACM}-{SIAM} Symposium on Discrete Algorithms},
  pages     = {1096--1115},
  year      = {2016},
  publisher = {Society for Industrial and Applied Mathematics},
  doi       = {10.1137/1.9781611974331.ch77},
  url       = {https://doi.org/10.1137/1.9781611974331.ch77},
  note      = {\url{https://doi.org/10.1137/1.9781611974331.ch77}}
}

@article{KratschWahlstrom2014,
  author  = {Stefan Kratsch and Magnus Wahlstr{\"o}m},
  title   = {Compression via matroids: A randomized polynomial kernel for {Odd Cycle Transversal}},
  journal = {ACM Transactions on Algorithms},
  volume  = {10},
  number  = {4},
  pages   = {20:1--20:15},
  year    = {2014},
  doi     = {10.1145/2635810},
  url     = {https://doi.org/10.1145/2635810},
  note    = {Article 20. \url{https://doi.org/10.1145/2635810}}
}

@misc{KoanaKumabe2026,
  author = {Tomohiro Koana and Soh Kumabe},
  title  = {A deterministic polynomial kernel for {Odd Cycle Transversal}},
  month  = {24 September},
  year   = {2026},
  url    = {https://arxiv.org/abs/2609.29141v1},
  note   = {Preprint, arXiv:2609.29141v1. \url{https://arxiv.org/abs/2609.29141v1}}
}

@inproceedings{RaviAgrawalKlein1991,
  author    = {R. Ravi and Ajit Agrawal and Philip Klein},
  title     = {Ordering problems approximated: Single-processor scheduling and interval graph completion},
  booktitle = {Automata, Languages and Programming},
  series    = {Lecture Notes in Computer Science},
  volume    = {510},
  pages     = {751--762},
  year      = {1991},
  doi       = {10.1007/3-540-54233-7_180},
  url       = {https://doi.org/10.1007/3-540-54233-7_180},
  note      = {\url{https://doi.org/10.1007/3-540-54233-7_180}}
}

@inproceedings{KashiwabaraFujisawa1979,
  author    = {Toshinobu Kashiwabara and Toshio Fujisawa},
  title     = {An {NP}-complete problem on interval graphs},
  booktitle = {Proceedings of the 1979 IEEE International Symposium on Circuits and Systems},
  pages     = {82--83},
  year      = {1979}
}

@article{FellowsEtAl2012,
  author  = {Michael R. Fellows and Fedor V. Fomin and Daniel Lokshtanov and Frances Rosamond and Saket Saurabh and Yngve Villanger},
  title   = {Local search: Is brute-force avoidable?},
  journal = {Journal of Computer and System Sciences},
  volume  = {78},
  number  = {3},
  pages   = {707--719},
  year    = {2012},
  doi     = {10.1016/j.jcss.2011.10.003},
  url     = {https://doi.org/10.1016/j.jcss.2011.10.003},
  note    = {\url{https://doi.org/10.1016/j.jcss.2011.10.003}}
}

@article{FulkersonGross1965,
  author  = {D. R. Fulkerson and O. A. Gross},
  title   = {Incidence matrices and interval graphs},
  journal = {Pacific Journal of Mathematics},
  volume  = {15},
  number  = {3},
  pages   = {835--855},
  year    = {1965},
  doi     = {10.2140/pjm.1965.15.835},
  url     = {https://doi.org/10.2140/pjm.1965.15.835},
  note    = {\url{https://doi.org/10.2140/pjm.1965.15.835}}
}

@article{NatanzonShamirSharan2001,
  author  = {Assaf Natanzon and Ron Shamir and Roded Sharan},
  title   = {Complexity classification of some edge modification problems},
  journal = {Discrete Applied Mathematics},
  volume  = {113},
  number  = {1},
  pages   = {109--128},
  year    = {2001},
  doi     = {10.1016/S0166-218X(00)00391-7},
  url     = {https://doi.org/10.1016/S0166-218X(00)00391-7},
  note    = {\url{https://doi.org/10.1016/S0166-218X(00)00391-7}}
}

@article{KratschWahlstrom2020,
  author  = {Stefan Kratsch and Magnus Wahlstr{\"o}m},
  title   = {Representative sets and irrelevant vertices: New tools for kernelization},
  journal = {Journal of the ACM},
  volume  = {67},
  number  = {3},
  pages   = {16:1--16:50},
  year    = {2020},
  doi     = {10.1145/3390887},
  url     = {https://doi.org/10.1145/3390887},
  note    = {Article 16. \url{https://doi.org/10.1145/3390887}}
}

@misc{Jansen2023KernelsAndBeyond,
  author       = {Bart M. P. Jansen},
  title        = {Kernels and Beyond},
  howpublished = {Interactive session at {FPT Fest in Honour of Mike Fellows}},
  year         = {2023},
  month        = jun,
  url          = {https://bjansen.win.tue.nl/talks.html},
  note         = {Bergen, Norway, 13 June 2023. \url{https://bjansen.win.tue.nl/talks.html}}
}

@article{Jansen2026QuoVadis,
  author  = {Bart M. P. Jansen},
  title   = {Quo Vadis, Kernelization?},
  journal = {Computer Science Review},
  volume  = {62},
  pages   = {101035},
  year    = {2026},
  doi     = {10.1016/j.cosrev.2026.101035},
  url     = {https://doi.org/10.1016/j.cosrev.2026.101035},
  note    = {\url{https://doi.org/10.1016/j.cosrev.2026.101035}}
}

@article{NatanzonShamirSharan2000,
  author  = {Assaf Natanzon and Ron Shamir and Roded Sharan},
  title   = {A polynomial approximation algorithm for the minimum fill-in problem},
  journal = {SIAM Journal on Computing},
  volume  = {30},
  number  = {4},
  pages   = {1067--1079},
  year    = {2000},
  doi     = {10.1137/S0097539798336073},
  url     = {https://doi.org/10.1137/S0097539798336073},
  note    = {\url{https://doi.org/10.1137/S0097539798336073}}
}
\endgroup

\end{document}